\documentclass[12pt]{iopart}
\usepackage[all]{xy}
\usepackage{amsthm,amssymb,nath}
\usepackage{graphicx}

\def\Includegraphics[#1]#2{{\tt[#1]#2}}

\let\epsilon\varepsilon

\def\eqref#1{{\rm(\ref{#1})}}
\newtheorem{proposition}{Proposition}

\newtheorem{definition}{Definition}

\newtheorem{remark}{Remark}

\def\dif{\mathop{}\!\mathrm d}
\def\ex{\mathrm e}

\begin{document}

\title{On multisoliton solutions of the constant astigmatism equation}
\author{Adam Hlav\'a\v{c}}
\address{Mathematical Institute in Opava, Silesian University in
  Opava, Na Rybn\'\i\v{c}ku 1, 746 01 Opava, Czech Republic.
  {\it E-mail}: Adam.Hlavac@math.slu.cz}
\date{}


\begin{abstract}
We introduce an algebraic formula producing infinitely many exact solutions of the constant astigmatism equation 
$ z_{yy} + ({1}/{z})_{xx} + 2 = 0 $
from a given seed. 
A construction of corresponding surfaces of constant astigmatism is then a matter of routine. 
As a special case, we consider multisoliton solutions of the constant astigmatism equation defined as counterparts of famous 
multisoliton solutions of the sine-Gordon equation. A few particular examples are surveyed as well.
\end{abstract}

\section{Introduction}
In this paper, we continue the investigation of 
the \it constant astigmatism equation (CAE) \rm 
$$\numbered\label{CAE}
z_{yy} + (\frac{1}{z})_{xx} + 2 = 0, 
$$
the Gauss equation of \it surfaces of constant astigmatism \rm immersed in Euclidean space. These surfaces are defined as having constant 
(but nonzero) difference between principal radii of curvature. 

Historical roots of the constant astigmatism surfaces can be traced back in works of the 19th century mathematicians 
\cite{Bia a, Bia b, Bia I, Bia II, Ri, Lip, vLil}, though the surfaces were nameless at that time. For more detailed history we refer to, 
e.g., \cite{P-S}. 

In 2009, after about a century in oblivion, the subject of constant astigmatism surfaces 
has been resurrected by Baran and Marvan in the work \cite{B-M I} concerning the systematic search for integrable classes of Weingarten surfaces. In the paper, the surfaces gained their name and Equation \eqref{CAE} was obtained as well. 
Recently, the Equation \eqref{CAE} has been examined by several authors \cite{H-M I, P-Z, H-M II, M-P, H-M III}. 

The main result of this paper is Proposition~\ref{prop1}, where
an algebraic formula \eqref{prop1eq} provides arbitrary many solutions of Equation \eqref{CAE} from a given seed. 
The proof is based on the observation that superposition formulas \cite[Eq.~21]{H-M I} (see also equations \eqref{supCAE} below) 
can be conveniently written in a matrix form and, therefore, their iteration coincides with the composition of linear transformations, 
i.e. matrix multiplication. Consequently, the same conclusion holds in the case of multisoliton solutions, 
i.e. solutions of the CAE having their counterparts in the well known multisoliton solutions of the sine-Gordon equation. 
Since the algebraic formula for the $n$-soliton sine-Gordon solution is well known, the corresponding solution of the CAE 
can be routinely computed using Proposition~\ref{propSoliton}. 

The results of Proposition~\ref{prop1} in combination with \cite[Prop.~3]{H-M I} also enable us 
to construct arbitrarily many constant astigmatism surfaces by purely algebraic manipulations and differentiation 
once an initial step (including an integration) is successfully performed.  

The most important results from the earliest history of the subject of constant astigmatism surfaces were, 
undoubtedly, obtained by Bianchi \cite{Bia a, Bia b, Bia I, Bia II}. 
He showed (see also \cite{Ri}) that evolutes (focal surfaces) of constant astigmatism surfaces are pseudospherical, i.e. with constant negative Gaussian curvature. Conversely, if one equips a pseudospherical surface with parabolic geodesic coordinates and takes the corresponding involutes, then they are of constant astigmatism. Bianchi also succeeded in finding some of the constant astigmatism surfaces explicitly \cite[Eq.~(30)]{Bia a}.    

A remarkable class of constant astigmatism surfaces was studied by Lipschitz \cite{Lip} 
and its subclass was later investigated by von~Lilienthal \cite{vLil}.   
Lipschitz parameterised his surfaces by spherical coordinates of the Gaussian image, see Figure~\ref{Lipsurf}.   
Recently, \cite{H-M II}, 
we showed that the solutions of the constant astigmatism equation that correspond
to the Lipschitz class of surfaces, are the Lie
symmetry invariant solutions and constitute a four-dimensional manifold. 
The counterpart sine-Gordon solutions are
shown to be Lie symmetry invariant as well. 
\begin{figure}[ht] 
\begin{center} 
\includegraphics[scale=0.17]{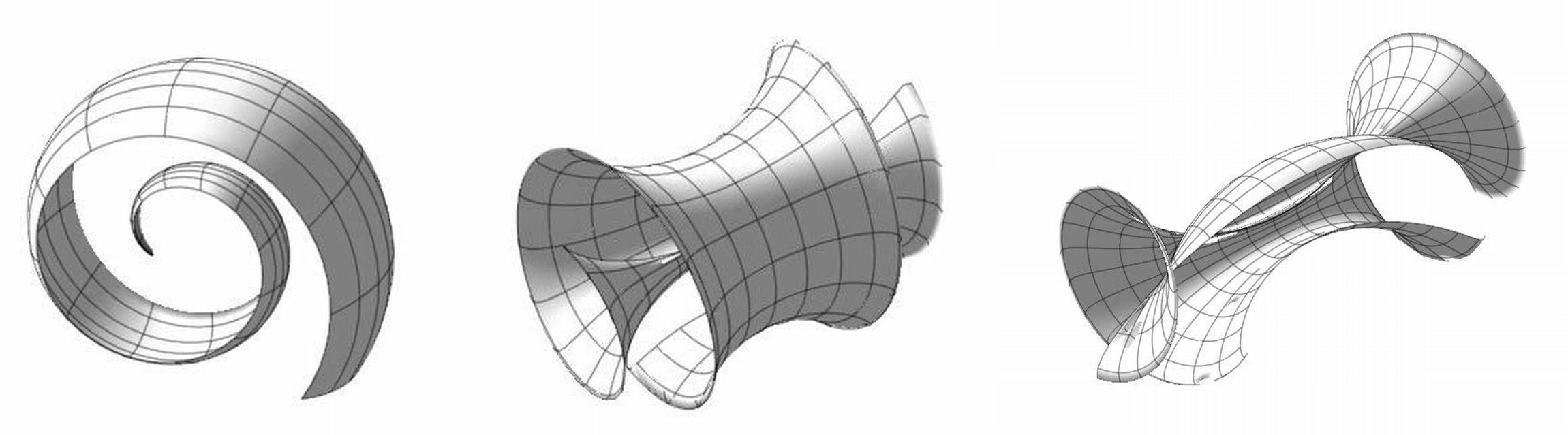}
\caption{Lipschitz surfaces of constant astigmatism.}\label{Lipsurf}
\end{center}
\end{figure} 

The aforementioned geometric link to pseudospherical surfaces served as a tool for 
deriving a nonlocal transformation between Equation \eqref{CAE} and the well-known sine-Gordon equation 
$$\numbered\label{SG}
\omega_{\xi\eta} = \frac12 \sin 2\omega,
$$
which is the Gauss equation of pseudospherical surfaces parameterised by asymptotic coordinates. For detailed description of the transformations see \cite{B-M I}.   

By an \it orthogonal equiareal pattern \rm (Sadowsky~\cite{Sad I,Sad II})  
we mean a parameterisation $x,y$ such that the corresponding metric is of the form
$$
z \mathrm{d} x^2 + \frac{1}{z} \mathrm{d} y^2,
$$
$z$ being an arbitrary function of $x,y$.
An associated {\it slip line field} is a parameterization $\xi,\eta$ such that the angle between
$\partial_x$ and $\partial_\xi$ as well as the angle between 
$\partial_y$ and $\partial_\eta$ is equal to $\pi/4$.
In \cite{H-M I} we observed that to every surface of constant astigmatism 
there corresponds an orthogonal equiareal pattern on the Gaussian sphere. Note that 
the same result was obtained by Bianchi~\cite[\S375, eq.~(20)]{Bia II} in the context of pseudospherical congruences.
We also showed that solutions of the sine-Gordon equation \eqref{SG} 
correspond to slip line fields on the same Gaussian sphere.
 
The method of \cite{H-M I} for generating solutions of the constant astigmatism equations and corresponding surfaces of constant astigmatism  has its origin in 
B\"acklund transformation for the sine-Gordon equation. 
Let $\omega$ be a solution of the sine-Gordon equation \eqref{SG}. Its B\"acklund transformation, 
$\omega^{(\lambda)}$,  
is given by the system
$$\numbered\label{BT}
\omega^{(\lambda)}_\xi  = \omega_\xi + \lambda\sin(\omega^{(\lambda)} + \omega), \qquad
\omega^{(\lambda)}_\eta = -\omega_\eta + \frac{1}{\lambda}\sin(\omega^{(\lambda)} - \omega),
$$
$\lambda$ being called a \it B\"acklund parameter. \rm
The famous Bianchi permutability theorem~\cite{Bia b}, see also \cite{Bia II, R-S}, 
and the \it superposition formula \rm
$$\numbered\label{supSG}
\tan\frac{\omega^{(\lambda_1\lambda_2)} - \omega}{2} = \frac{\lambda_1 + \lambda_2}{\lambda_1 - \lambda_2} \tan\frac{\omega^{(\lambda_1)} - \omega^{(\lambda_2)}}{2}
$$ 
enable us to compute solution $\omega^{(\lambda_1\lambda_2)}$, the B\"acklund transformation of $\omega^{(\lambda_1)}$
with B\"acklund parameter $\lambda_2$, by algebraic manipulations. 
Therefore, when a seed solution $\omega$ is given, the integration of the system \eqref{BT} needs to be done only once and the further iteration of the B\"acklund transformation can be performed purely algebraically.

Pairs of sine-Gordon solutions related by B\"acklund transformation induce solutions of the constant astigmatism equation.
In \cite[Def.~4]{H-M I} we defined the \it associated potentials \rm $f^{(\lambda)}, x^{(\lambda)}, y^{(\lambda)} $
corresponding to a pair of sine-Gordon solutions $\omega, \omega^{(\lambda)}$. In terms of
$
g^{(\lambda)} = `e^{f^{(\lambda)}}
$
the associated potentials are given by 
$$\numbered\label{aspot}
g^{(\lambda)}_\xi = g^{(\lambda)}\lambda \cos(\omega^{(\lambda)} + \omega),
\qquad 
g^{(\lambda)}_\eta = g^{(\lambda)}\frac1\lambda \cos(\omega^{(\lambda)} - \omega),
\\
x^{(\lambda)}_\xi = \lambda g^{(\lambda)} \sin(\omega^{(\lambda)} + \omega),
\qquad
x^{(\lambda)}_\eta = \frac1\lambda g^{(\lambda)} \sin(\omega^{(\lambda)} - \omega),
\\
y^{(\lambda)}_\xi = \frac{\lambda \sin(\omega^{(\lambda)} + \omega)}{g^{(\lambda)}}  ,
\qquad
y^{(\lambda)}_\eta = - \frac{\sin(\omega^{(\lambda)} - \omega)}{\lambda g^{(\lambda)}} .
$$
Expressing $z^{(\lambda)} = 1/{g^{(\lambda)}}^2 $ in terms of $x^{(\lambda)}$ and $y^{(\lambda)}$ one obtains a solution of the CAE. 

In the same paper (Prop. 5) we succeeded in extending the superposition principle~\eqref{supSG} 
to solutions of the CAE. Let $\omega, \omega^{(\lambda_1)}, \omega^{(\lambda_2)}, 
 \omega^{(\lambda_1\lambda_2)}$ be four sine-Gordon solutions  
related by the Bianchi superposition principle~\eqref{supSG}.
Then the associated potentials 
$g^{(\lambda_1\lambda_2)}$, $x^{(\lambda_1\lambda_2)}$, $y^{(\lambda_1\lambda_2)}$
corresponding to the pair $\omega^{(\lambda_1)},\omega^{(\lambda_1\lambda_2)}$
are related to the associated potentials 
$g^{(\lambda_2)}$, $x^{(\lambda_2)}$, $y^{(\lambda_2)}$ corresponding to the pair 
$\omega,\omega^{(\lambda_2)}$ by formulas
$$\numbered\label{supCAE}
g^{(\lambda_1\lambda_2)} = 
\frac{-\lambda_1\lambda_2}{\lambda_1^2+\lambda_2^2-2\lambda_1\lambda_2\cos(\omega^{(\lambda_1)} - \omega^{(\lambda_2)})}g^{(\lambda_2)},
\\
x^{(\lambda_1\lambda_2)}
 = \frac{\lambda_1 \lambda_2}{\lambda_1^2 - \lambda_2^2} 
   (x^{(\lambda_2)}
     -  
      \frac{2 \lambda_1 \lambda_2 \sin(\omega^{(\lambda_1)} - \omega^{(\lambda_2)})}
         {\lambda_1^2 + \lambda_2^2
          - 2 \lambda_1 \lambda_2 \cos(\omega^{(\lambda_1)} - \omega^{(\lambda_2)})
          } g^{(\lambda_2)}),
\\
y^{(\lambda_1\lambda_2)}
 = \frac{\lambda_1^2 - \lambda_2^2}{\lambda_1\lambda_2} y^{(\lambda_2)}
   - \frac{2 \sin(\omega^{(\lambda_1)} - \omega^{(\lambda_2)})}{g^{(\lambda_2)}}.
$$
The above formulas can be regarded as a starting point of our approach in this paper.

The contents of this paper are as follows. 
In Section~2 we observe that equations~\eqref{supCAE} can be conveniently
exposed in matrix form and the iteration of the superposition principle for the CAE reduces to mere matrix multiplication. 
In Section~3 we handle the multisoliton case, i.e. we deal with solutions of the CAE coming from the well known $n$-soliton solutions of the sine-Gordon equation. 
Section~4 deals with constant astigmatism surfaces and slip line fields. 
In Section~5 we study in detail the $n$-soliton case for $n = 1,2,3$ producing exact solutions of the CAE, constant astigmatism surfaces and slip line fields, whilst
the last section is devoted to the subcase when all B\"acklund parameters are equal to $1$.    

\section{Solutions of the CAE}
Let $\omega^{[0]} = \bar{\omega}^{[0]} $ be some seed solution of the sine-Gordon equation. 
Fix B\"acklund parameters $\lambda_1, \dots, \lambda_{k + 1}$ and let us denote
$$\numbered\label{defk}
\omega^{[k]} = \omega^{(\lambda_1\lambda_2\ldots\lambda_k)}, \qquad
\bar{\omega}^{[k]} = \omega^{(\lambda_2\lambda_3\ldots\lambda_{k+1})},
$$
see the diagram (a part of the well known Bianchi lattice)
\begin{displaymath} \numbered \label{lattice}
    \xymatrix{ 
\omega^{[0]} \ar[r]^{\lambda_2} \ar[d]_{\lambda_1} 	& \bar{\omega}^{[1]} \ar[r]^{\lambda_3} \ar[d]_{\lambda_1} & \bar{\omega}^{[2]} \ar[d]_{\lambda_1} \ar[r]^{\lambda_4}  & \bar{\omega}^{[3]} \ar[d]_{\lambda_1} \ar[r]^{\lambda_5}  & \bar{\omega}^{[4]} \ar[d]_{\lambda_1} \ar[r]^{\lambda_6} & \dots
\\
\omega^{[1]} \ar[r]^{\lambda_2}  & 	\omega^{[2]}  \ar[r]^{\lambda_3} & 	\omega^{[3]} \ar[r]^{\lambda_4} & 	\omega^{[4]} \ar[r]^{\lambda_5} & 	\omega^{[5]} \ar[r]^{\lambda_6} & \dots
}
\end{displaymath}
Let $g^{[j]}, x^{[j]}, y^{[j]}$ denote the associated potentials \eqref{aspot} corresponding to the pair $\bar{\omega}^{[j-1]}, \omega^{[j]}$. 
In this notation, the superposition formulas \eqref{supCAE} 
turn out to be 
$$\numbered\label{difeq}
x^{[j+1]}
 = \frac{\lambda_{j+1} \lambda_1}{\lambda_{j+1}^2 - \lambda_1^2} 
   (x^{[j]}
     -  
      \frac{2 \lambda_{j+1} \lambda_1  \sin(\bar{\omega}^{[j]} - \omega^{[j]})}
         {\lambda_{j+1}^2 + \lambda_1^2
          - 2 \lambda_{j+1} \lambda_1 \cos(\bar{\omega}^{[j]} - \omega^{[j]})
          } g^{[j]} ),
\\
y^{[j+1]}
 = \frac{\lambda_{j+1}^2 - \lambda_1^2}{\lambda_{j+1}\lambda_1} y^{[j]}
   - \frac{2 \sin(\bar{\omega}^{[j]} - \omega^{[j]})}{g^{[j]}},
\\
g^{[j+1]} = \frac{-\lambda_{j+1}\lambda_1}{\lambda_{j+1}^2+\lambda_1^2-2\lambda_{j+1}\lambda_1\cos(\bar{\omega}^{[j]} - \omega^{[j]})}g^{[j]}.
$$
The above formulas constitute recurrence relations for the quantities $x^{[n]}, y^{[n]}, g^{[n]}$ with the initial conditions
$$ \numbered\label{ic}
x^{[1]} = x_1, \quad y^{[1]} = y_1, \quad g^{[1]} = g_1. 
$$
\begin{proposition}\label{prop1}
Let $x_1, y_1, g_1$ be the associated potentials corresponding to the pair $\omega^{[0]}, \omega^{[1]}$ of sine-Gordon solutions.
Let $S^{[j]}$ be $4\times 4$ matrices with entries defined by formulas 
$$\numbered\label{Smatrix}
S^{[j]}_{11} = \frac{\lambda_{j+1} \lambda_1}{\lambda_{j+1}^2 - \lambda_1^2}, \quad
S^{[j]}_{13} = -\frac{\lambda_{j+1}^2 \lambda_1^2}{\lambda_{j + 1}^2 - \lambda_1^2} \cdot \frac{2  \sin(\bar{\omega}^{[j]} - \omega^{[j]})}
{\lambda_{j+1}^2 + \lambda_1^2 - 2 \lambda_{j+1} \lambda_1 \cos(\bar{\omega}^{[j]} - \omega^{[j]})}, \\
S^{[j]}_{22} = \frac{\lambda_{j+1}^2 - \lambda_1^2}{\lambda_{j+1}\lambda_1}, \quad
S^{[j]}_{24} = - 2 \sin(\bar{\omega}^{[j]} - \omega^{[j]}), \\
S^{[j]}_{33} = \frac{1}{S^{[j]}_{44}} = \frac{-\lambda_{j+1}\lambda_1}{\lambda_{j+1}^2+\lambda_1^2-2\lambda_{j+1}\lambda_1\cos(\bar{\omega}^{[j]} - \omega^{[j]})} 
$$
all the other entries being zero. Let 
$$\numbered\label{prop1eq}
(
\begin{array}{c}
x^{[n]} \\
y^{[n]} \\
g^{[n]} \\
1/g^{[n]} 
\end{array}
)
=
\prod_{i = 1}^{n - 1} S^{[i]} 
(
\begin{array}{c}
x_1 \\
y_1 \\
g_1 \\
1/g_1 
\end{array}
).
$$
Then $x^{[n]}, y^{[n]}, g^{[n]}$ are the associated potentials corresponding to the pair $\bar{\omega}^{[n-1]}, \omega^{[n]}$. Moreover, 
if $z^{[n]} = 1/{g^{[n]}}^2 $, then $z^{[n]}(x^{[n]},y^{[n]})$ is a solution of the constant astigmatism equation \eqref{CAE}.
\end{proposition}
\begin{proof}
One easily observes that relations \eqref{difeq}  
can be obtained by acting of $4\times 4$ matrix $S^{[j]}$ with the only nonzero entries \eqref{Smatrix} 
on the column vector $X^{[j]} = (x^{[j]}, y^{[j]},g^{[j]}, 1/g^{[j]})$. The recurrent formula \eqref{difeq} with initial conditions 
\eqref{ic} then can be conveniently written as 
$$
X^{[j + 1]} = S^{[j]}X^{[j]}, \qquad X^{[1]} = X_1 = (x_1, y_1,g_1, 1/g_1)^T.
$$
Hence
$$
X^{[n]} =  \prod_{i = 1}^{n - 1} S^{[i]} X_1.
$$
\end{proof}

\begin{remark} \rm
Expanding the matrix formula \eqref{prop1eq} we have
$$
x^{[n]} = x_1 \lambda_1^{n-1}
\prod_{i=1}^{n-1}\frac {\lambda_{i+1}}{
\lambda_{i+1}^{2}-\lambda_{1}^{2}} \\
-
2\sum _{i=1}^{n-1}
 \left( \frac{
\lambda_{i+1}\lambda_1^{n-i+1} g^{[i]} \sin ( \bar{\omega}^{[i]}-\omega^{[i]} ) }
 { \lambda_{i+1}^{2} + \lambda_1^{2}-2\lambda_{i+1}\lambda_1\cos ( \bar{\omega}^{[i]}-\omega^{[i]} ) 
}
\prod _{p=1}^{n-i}{\frac{\lambda_{n-
p+1}}{\lambda_{{n-p+1}}^{2}-\lambda_1^2}}
 \right), 
\\
y^{[n]} = 
\frac{y_1}{\lambda_1^{n-1}}
\prod _{i=1}^{n-1}{\frac {\lambda_{{i+1}}^{2}-\lambda_{{1}}^{2}}{
\lambda_{{i+1}}}}
-
2\sum _{i=1}^{n-1} 
(
\frac{\sin ( \bar{\omega}^{[i]}-\omega^{[i]} )}
{\lambda_{{1}}^{n-i-1} g^{[i]} }
\prod _{p=1}^{n-i-1}{\frac {\lambda_{{n-p
+1}}^{2}-\lambda_{{1}}^{2}}{\lambda_{{n-p+1}}}}
),
\\
g^{[n]} = \lambda_1^{n-1} g_1 \prod _{i=1}^{n-1} \frac{-\lambda_{i+1}}{\lambda_{i+1}^2+\lambda_1^2-2\lambda_{i+1}\lambda_1\cos(\bar{\omega}^{[i]} - \omega^{[i]})}.
$$
\end{remark}

\section{Multisoliton solutions} 
Let $\omega^{[0]} = 0$. Let us define $\lambda_{kl}^+ := \lambda_k + \lambda_l,\;\; \lambda_{kl}^- := \lambda_k - \lambda_l$ and 
$$
a_i := \exp(\lambda_i \xi + \frac{\eta}{\lambda_i} + c_i).
$$
Solving System \eqref{BT}, we get one-soliton solutions  
$$
\omega^{[1]} = 2\arctan a_1, \qquad 
\bar{\omega}^{[1]} = 2\arctan a_2
$$
and, applying the superposition principle \eqref{supSG} to the triple $\omega^{[0]}, \omega^{[1]}, \bar{\omega}^{[1]}$, 
we easily obtain the two-soliton solutions
$$\numbered\label{twosol}
\omega^{[2]} = 2\arctan\frac{\lambda_{12}^+ (a_1-a_2)}{\lambda_{12}^- (1+a_1 a_2)}, \qquad 
\bar{\omega}^{[2]} = 2\arctan\frac{\lambda_{23}^+(a_2-a_3)}{\lambda_{23}^-(1+a_2a_3)}.
$$
An exact analytic $n$-soliton solution, in our notation $\omega^{[n]}$, of the sine-Gordon equation has been obtained by
several authors \cite{Abl, Bry, Cau I, Cau II, Cau III, Hir}, see also \cite{A-I, S-B}.
The formula best suited for this paper can be found e.g. in \cite{Cau III} and is of the form 
$$\numbered\label{SGmultisol}
 \omega^{[n]} = 
\frac12 \arccos  \varphi^{[n]},
$$
where 
$$\numbered\label{phi}
\varphi^{[n]} = 1 - 2\frac{\partial^2}{\partial \xi \,\partial\eta} \ln \det M
$$
$M$ being the $n\times n$ matrix with entries 
$$
M_{ij} = \frac{1}{\lambda_i + \lambda_j}
(\sqrt{a_i a_j} + \frac{1}{\sqrt{a_i a_j}}).
$$
Note also that $\bar{\omega}^{[n]}$ arises from $\omega^{[n]}$ by increasing all lambdas' indices by one, namely
$$\numbered\label{SGmultisolbar}
\bar{\omega}^{[n]} = 
 \frac12 \arccos \bar{\varphi}^{[n]},
$$
where
$$\numbered\label{phibar}
\bar{\varphi}^{[n]} = 1 - 2\frac{\partial^2}{\partial \xi \, \partial\eta} \ln \det \bar{M}
$$
$\bar{M}$ being given by
$$
\bar{M}_{ij} = \frac{1}{\lambda_{i + 1} + \lambda_{j + 1}}
(\sqrt{a_{i+1}a_{j+1}} + \frac{1}{\sqrt{a_{i+1}a_{j+1}}}).
$$

\begin{definition} \rm
By a \it $j$-soliton solution of the constant astigmatism equation \rm 
we shall mean a triple $(x^{[j]}, y^{[j]}, g^{[j]})$ formed by associated potentials corresponding to the 
$j$-soliton solution $\omega^{[j]}$ and the $(j-1)$-soliton solution $\bar{\omega}^{[j-1]}$ (see Diagram \eqref{lattice}) of the sine-Gordon equation.
\end{definition}
\begin{remark} \rm
To obtain a solution of the CAE explicitly, one would have to express $z^{[j]} = 1/{g^{[j]}}^2$ in terms of $x^{[j]}$ and $y^{[j]}$.
However, this is almost never possible in terms of elementary functions, see examples in Section~5 and~6.
\end{remark}

A one-soliton solution of the CAE is easy to construct. 
Following \cite{Bia a, Bia II}, see also \cite[Prop. 4]{H-M I}, $x^{[1]} = x_1$ and $g^{[1]} = g_1$ can be obtained by differentiation, namely 
$$\numbered\label{CA1solxg}
g_1 = \frac{\dif \omega^{[1]}}{\dif c_1} = \frac{2 a_1}{1+a_1^2}, \qquad
x_1 = -\frac{\dif \ln g_1}{\dif c_1} = \frac{a_1^2-1}{a_1^2 + 1}.
$$
For $y^{[1]} = y_1$ we have the system 
$$
y^{[1]}_\xi = \frac{\lambda_1\sin(\omega^{[1]} + \bar{\omega}^{[0]})}{g_1} = \lambda_1,  \\
y^{[1]}_\eta = -\frac{\sin(\omega^{[1]} - \bar{\omega}^{[0]})}{\lambda_1 g_1} = -\frac{1}{\lambda_1}
$$
with the general solution 
$$\numbered\label{CA1soly}
y_1 = \lambda_1 \xi - \frac{\eta}{\lambda_1} + k_1, 
$$
$k_1$ being an arbitrary constant. 
Setting $z_1 = 1/g_1^2$, eliminating $\xi,\eta$ and dropping the lower indices, one reveals the \it von~Lilienthal solution \rm
$$\numbered\label{vLsol}
z = \frac{1}{1 - x^2},
$$
see Figure \ref{onesol} on p.~11.

\begin{proposition} \label{propSoliton}
Let us denote 
$A^{[j]} = 2 \bar{\varphi}^{[j]} \varphi^{[j]}$ and 
$B^{[j]} = 2 \sqrt{({\bar{\varphi}^{[j]}}^2 - 1) ({\varphi^{[j]}}^2 - 1)}$, where 
$\varphi^{[j]}$ and $\bar{\varphi}^{[j]}$ are defined by \eqref{phi} and \eqref{phibar} respectively.
Then the $n$-soliton solution of the CAE is given by the formula
$$\numbered\label{nsolCAE}
(
\begin{array}{c}
x^{[n]} \\
y^{[n]} \\
g^{[n]} \\
1/g^{[n]} 
\end{array}
)
=
\prod_{i = 1}^{n - 1} S^{[i]} 
(
\begin{array}{c}
x_1 \\
y_1 \\
g_1 \\
1/g_1 
\end{array}
),
$$
where the only nonzero entries of matrices $S^{[j]}$ are given by 
$$\numbered\label{Smatrixsol}
S^{[j]}_{11} = \frac{\lambda_{j+1} \lambda_1}{\lambda_{j+1}^2 - \lambda_1^2}, \qquad 
S^{[j]}_{13} =  \frac{\lambda_{j+1}^2 \lambda_1^2}{\lambda_1^2 - \lambda_{j+1}^2} \cdot \frac
{\sqrt{2-A^{[j]}-B^{[j]}}}
{\lambda_1^2+\lambda_{j+1}^2 - \lambda_{j+1}\lambda_1 \sqrt{2+ A^{[j]} + B^{[j]}} } \\
S^{[j]}_{22} = \frac{\lambda_{j+1}^2 - \lambda_1^2}{\lambda_{j+1}\lambda_1}, \qquad
S^{[j]}_{24} = -\sqrt{2-A^{[j]}-B^{[j]}}, \\
S^{[j]}_{33} = \frac{1}{S^{[j]}_{44}} = 
\frac
{-\lambda_{j+1} \lambda_1}
{\lambda_1^2 + \lambda_{j+1}^2 -  \lambda_{j+1} \lambda_1 \sqrt{2+A^{[j]}+B^{[j]}} }.
$$ 
\end{proposition}

\begin{proof}
Formulas \eqref{Smatrixsol} follow from plugging \eqref{SGmultisol} and \eqref{SGmultisolbar} into \eqref{Smatrix} 
and employing trigonometric identities.
\end{proof}

\section{Surfaces of constant astigmatism}
Let $\mathbf r $ be a pseudospherical surface corresponding to a sine-Gordon solution $\omega$. Its B\"acklund transformation $\mathbf r^{(\lambda)}$ is given by the formula 
$$\numbered\label{BTsurf}
\mathbf r^{(\lambda)} = \mathbf r
 + \frac{2\lambda}{1 + \lambda^2}(\frac{\sin(\omega - \omega^{(\lambda)})}{\sin(2\omega)} {\mathbf  r_{\xi}}
 + \frac{\sin(\omega + \omega^{(\lambda)})}{\sin(2\omega)} {\mathbf r_{\eta}}).
$$
Let us define (cf. \eqref{defk})
$$
\mathbf r^{[k]} = \mathbf r^{(\lambda_1\lambda_2\ldots\lambda_k)}, \qquad
\bar{\mathbf r}^{[k]} = \mathbf r^{(\lambda_2\lambda_3\ldots\lambda_{k+1})}.
$$
Then we have the recurrence relation
$$\numbered\label{BTsurfseq}
\mathbf r^{[j+1]} = 
\mathbf r^{[j]}
+ \frac{2\lambda_{j+1}}{1 + \lambda_{j+1}^2}(\frac{\sin(\omega^{[j]} - \omega^{[j+1]})}{\sin(2\omega^{[j]})} {\mathbf  r^{[j]}_{\xi}}
+ \frac{\sin(\omega^{[j]} + \omega^{[j+1]})}{\sin(2\omega^{[j]})} {\mathbf r^{[j]}_{\eta}})
$$
with the initial condition $\mathbf r^{[0]} = \bar{\mathbf r}^{[0]} = \mathbf r_0$. Surfaces $\bar{\mathbf r}^{[i]} $ are obtained from $\mathbf r^{[i]} $ simply by increasing all lambdas' indices by one and replacing 
$\omega^{[i]}$ with $\bar{\omega}^{[i]}$. The iteration process is shown in the diagram below, cf. \eqref{lattice}.
\begin{displaymath}
    \xymatrix{ 
{\bf r}^{[0]} \ar[r]^{\lambda_2} \ar[d]_{\lambda_1} 	& \bar{{\bf r}}^{[1]} \ar[r]^{\lambda_3} \ar[d]_{\lambda_1} & \bar{{\bf r}}^{[2]} \ar[d]_{\lambda_1} \ar[r]^{\lambda_4}  & \bar{{\bf r}}^{[3]} \ar[d]_{\lambda_1} \ar[r]^{\lambda_5}  & \bar{{\bf r}}^{[4]} \ar[d]_{\lambda_1} \ar[r]^{\lambda_6} & \dots
\\
{\bf r}^{[1]} \ar[r]^{\lambda_2}  & 	{\bf r}^{[2]}  \ar[r]^{\lambda_3} & 	{\bf r}^{[3]} \ar[r]^{\lambda_4} & 	{\bf r}^{[4]} \ar[r]^{\lambda_5} & 	{\bf r}^{[5]} \ar[r]^{\lambda_6} & \dots
}
\end{displaymath} 
Substituting $\lambda = 1$ into \eqref{BTsurf} one gets what is called a \it complementary \rm pseudospherical surface
$$
\mathbf r^{(1)} = \mathbf r
 + \frac{\sin(\omega - \omega^{(1)})}{\sin(2\omega)} {\mathbf  r_{\xi}}
 + \frac{\sin(\omega + \omega^{(1)})}{\sin(2\omega)} {\mathbf r_{\eta}}.
$$
Obviously, the surfaces $\mathbf r^{[j]}$ and $\bar{\mathbf r}^{[j-1]}$ become complementary when substituting 
$\lambda_1 = 1$ into~$\mathbf r^{[j]}$.

According to \cite[Prop. 3]{H-M I}, the common involute, $\tilde{\bf r}^{[j]}$, 
of a pair of complementary pseudospherical surfaces, $ \mathbf r^{[j]}|_{\lambda_1 = 1} $ and $ \bar{\mathbf r}^{[j-1]} $, 
is of constant astigmatism and is given by the equation
$$
\tilde{\bf r}^{[j]} =  \bar{\mathbf r}^{[j-1]} - \tilde{\mathbf n}^{[j]} \ln g^{[j]} |_{\lambda_1 = 1}, 
$$
where $g^{[j]} $ is determined by \eqref{prop1eq} and  $\tilde{\mathbf n}^{[j]}$, a unit normal of the constant astigmatism surface, is simply
$$
\tilde{\mathbf n}^{[j]} = \left. \mathbf r^{[j]} \right|_{\lambda_1 = 1} - \bar{\mathbf r}^{[j-1]} .
$$
If the surfaces $\mathbf r^{[j]}|_{\lambda_1 = 1} $ and $ \bar{\mathbf r}^{[j-1]} $ are $j$-soliton and $(j-1)$-soliton pseudospherical surfaces respectively, then the corresponding common involute, $\tilde{\bf r}^{[j]}$, will be called a \it $j$-soliton surface of constant astigmatism. \rm
 
Let us also remark that $\tilde{\mathbf n}^{[j]}(\xi,\eta)$ parameterises a unit sphere by slip lines (see the Introduction, 
for details see \cite{H-M I}).

\section{Examples of multisoliton solutions} 
In this section we provide explicit formulas for some multisoliton solutions of the constant astigmatism equation 
as well as corresponding constant astigmatism surfaces. 
 
Firstly, let us introduce a notation. Let us define 
$$
\alpha := \xi - \eta, \qquad \beta := \xi + \eta
$$
which is nothing but the space and time coordinates in which sine-Gordon equation is of the form 
$\omega_{\beta\beta} - \omega_{\alpha\alpha}  = \sin \omega$. 
Also recall that
$$
a_i := \exp(\lambda_i \xi + \frac{\eta}{\lambda_i} + c_i), \qquad
\lambda_{kl}^+ := \lambda_k + \lambda_l,\qquad \lambda_{kl}^- := \lambda_k - \lambda_l
$$
and, in order to have short formulas, let us define 
$$
a := \exp( \xi + \eta + c), \qquad
\lambda_{kl}^\oplus := \lambda_k^2 + \lambda_l^2,\qquad 
\lambda_{kl}^\ominus := \lambda_k^2 - \lambda_l^2, \\
\lambda_{k}^+ := \lambda_k + 1,\qquad \lambda_{k}^- := \lambda_k -1, \qquad 
\lambda_{k}^\oplus := \lambda_k^2 + 1,\qquad \lambda_{k}^\ominus := \lambda_k^2 -1.
$$

\subsection{One-soliton solutions}
In Section 3 we constructed the one soliton solution of the CAE corresponding to the pair $\omega^{[0]} = 0$ and 
$\omega^{[1]} = 2\arctan a_1$. It belongs to the von~Lilienthal class. 
Let us proceed to the corresponding surfaces of constant astigmatism.
The family of well known one-soliton Dini's surfaces is
$$
\bar{\mathbf r}^{[1]} = 
(
\frac{4 \lambda_2 a_2 \cos\alpha }{\lambda_2^+ (a_2^2 + 1)}, 
\frac{4 \lambda_2 a_2 \sin\alpha }{\lambda_2^+ (a_2^2 + 1)},
\beta-\frac{2 \lambda_2 (a_2^2-1)}{\lambda_2^+ (a_2^2 + 1)}
).
$$
Substituting $\lambda_1 = 1$ into another Dini's surface
$$ 
\mathbf r^{[1]} = (
\frac{4 \lambda_1 a_1 \cos\alpha }{\lambda_1^+ (a_1^2 + 1)}, 
\frac{4 \lambda_1 a_1 \sin\alpha }{\lambda_1^+ (a_1^2 + 1)},
\beta-\frac{2 \lambda_1 (a_1^2-1)}{\lambda_1^+ (a_1^2 + 1)}
)
$$ 
we obtain the pseudosphere 
$$\numbered\label{pseudosphere}
\mathbf r^{[1]}|_{\lambda_1 = 1} = 
(
\frac{2 a \cos\alpha }{a^2 + 1}, 
\frac{2 a \sin\alpha }{a^2 + 1},
\beta+ \frac{1 - a^2}{1 + a^2}
).
$$
Note that in this case the `seed surface' $\mathbf r^{[0]} = \mathbf r_0$ is degenerated and coincides with the $z$-axis $(0,0,\beta)$.

The Gaussian map of corresponding constant astigmatism surface is 
$$
\tilde{\mathbf{n}}^{[1]} = \mathbf r^{[1]}|_{\lambda_1 = 1} - \mathbf{r}_0 = 
( \frac{2a \cos\alpha}{a^2 + 1}, 
\frac{2a \sin\alpha}{a^2 + 1},
\frac{1 - a^2}{1 + a^2}
),
$$
forming the net of $45^{\circ}$ loxodromes on the unit sphere. The family of one-soliton constant astigmatism surfaces is then  
$$\numbered\label{vLsurf}
\tilde{\mathbf{r}}^{[1]} = \mathbf{r}_0 -  \tilde{\mathbf{n}}^{[1]} \ln (k g^{[1]})|_{\lambda_1 = 1} = (
\begin{array}{c}
\frac{-2a}{1+a^2} \ln(\frac{2k a}{1+a^2})\cos\alpha \\
\frac{-2a}{1+a^2} \ln(\frac{2k a}{1+a^2})\sin\alpha \\
\beta + \frac{a^2-1}{a^2 + 1}\ln(\frac{2k a}{1+a^2}) 
\end{array}
),
$$ 
$k$ being a real constant.
The surfaces coincide with the von~Lilienthal class \cite{vLil}, see Figure \ref{onesol}; for detailed description and pictures see \cite{B-M I}.  
Evolutes of the surface $\tilde{\mathbf{r}}^{[1]}$
are the pseudosphere $\mathbf{r}^{[1]}|_{\lambda_1 = 1}$ and the $z$-axis $ \mathbf{r}_0$. 

\begin{figure}[ht] 
\begin{center} 
\includegraphics[scale=0.25]{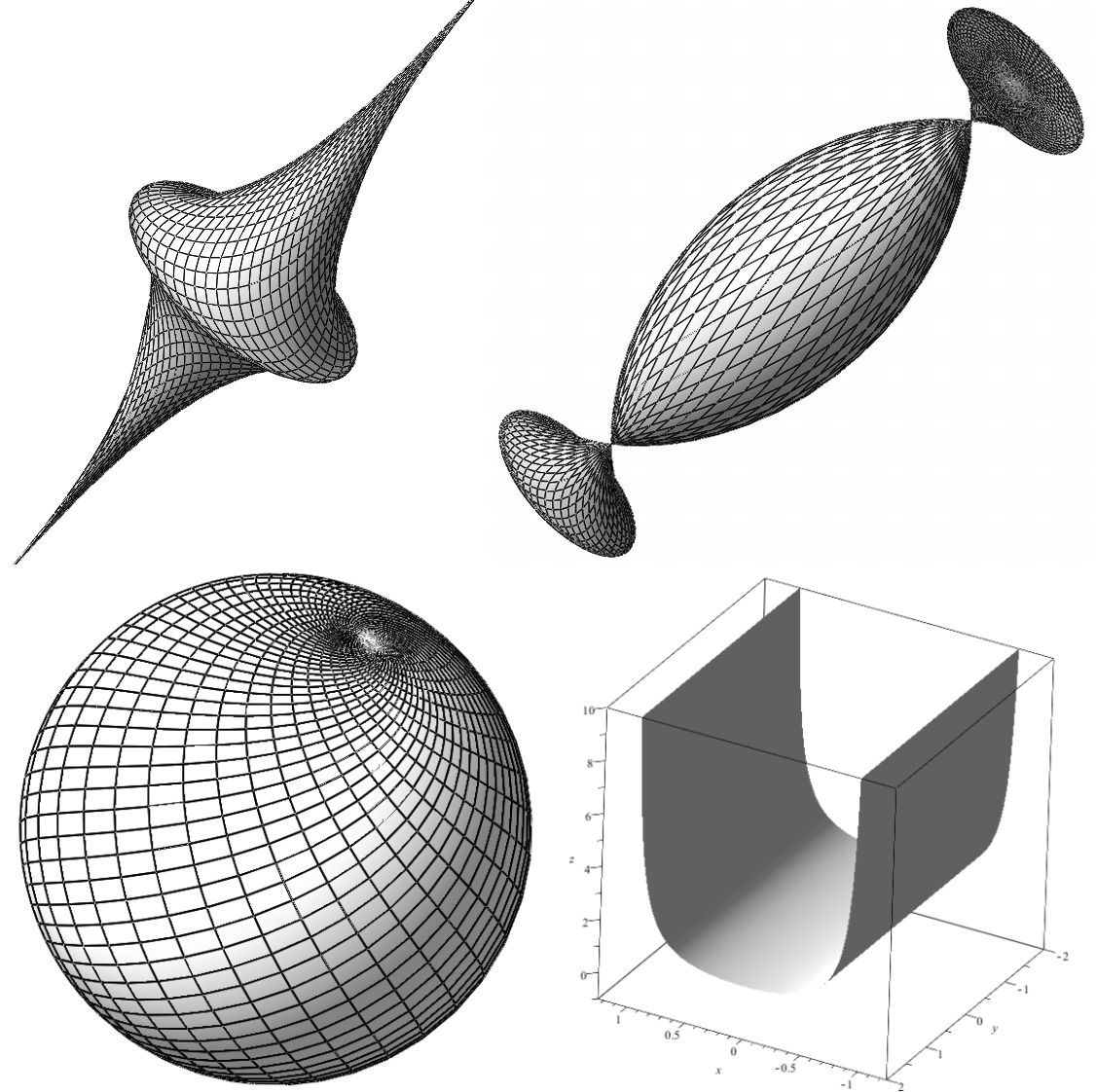}
\caption{Pseudosphere $\mathbf r^{[1]}(\xi,\eta)|_{\lambda_1 = 1}$ (upper left); 
von~Lilienthal surface $\tilde{\mathbf{r}}^{[1]}(\xi,\eta)$ for $\ln k = 0.3$ (upper right);
Gaussian map $\tilde{\mathbf{n}}^{[1]}(\xi,\eta)$ (lower left);
solution $z = 1/(1-x^2)$ of the CAE (lower right).}\label{onesol} 
\end{center}
\end{figure} 

\subsection{Two-soliton solutions}\label{TwosolEx}
The two-soliton solutions of the sine-Gordon equation are given by \eqref{twosol}
and the two soliton solution of the CAE corresponding to the pair  $\bar{\omega}^{[1]}$ and $\omega^{[2]}$ 
is 
$$\numbered\label{CA2sol}
x^{[2]} = \frac{-\lambda_2  \lambda_1}{{\lambda_{12}^\ominus}} \cdot 
\frac{
  {\lambda_{12}^+}^2 (a_1^2-a_2^2) + {\lambda_{12}^-}^2 (a_1^2 a_2^2 - 1)}
{   
{\lambda_{12}^+}^2 (a_1^2+a_2^2) 
+  {\lambda_{12}^-}^2 (a_1^2 a_2^2+1)
-8 {\lambda_1} {\lambda_2} a_1 a_2 }, \\
y^{[2]} =
-\frac{\lambda_{12}^\ominus}{\lambda_1^2\lambda_2} (\lambda_1^2 \xi - \eta)
+
\frac{2 (1+a_1 a_2)(a_1-a_2)}{a_1 (1+a_2^2)} , \\
{g^{[2]}} = 
\frac
{-2 \lambda_1 \lambda_2  a_1 (1+a_2^2)}
{{\lambda_{12}^+}^2 (a_1^2+a_2^2) 
+  {\lambda_{12}^-}^2 (a_1^2 a_2^2+1)
-8 {\lambda_1} {\lambda_2} a_1 a_2}, \\
z^{[2]} = \frac{1}{{g^{[2]}}^2} = (
\frac{{\lambda_{12}^+}^2 (a_1^2+a_2^2) 
+  {\lambda_{12}^-}^2 (a_1^2 a_2^2+1)
-8 {\lambda_1} {\lambda_2} a_1 a_2 }
{2 \lambda_1 \lambda_2  a_1 (1+a_2^2)})^2.
$$
The triple $z^{[2]}$, $x^{[2]}$, $y^{[2]}$ provides a solution of the CAE in parametric form; the plot 
can be seen in Figure \ref{CAsol}. 
Eliminating $\xi, \eta$ one obtains an implicit formula for the function $z(x,y) = z^{[2]}(x^{[2]},y^{[2]})$, namely
$$
y =
2 \ln a_2  - 
\frac{\lambda_{12}^\oplus \ln a_1}{\lambda_1  \lambda_2 }
+
\frac{2 (1+a_1 a_2)(a_1-a_2)}{a_1 (1+a_2^2)} ,
$$
where 
$$
a_1 = \frac{
-(x^2{\lambda_{12}^{\ominus}}^2 - {\lambda_{1}^2} {\lambda_{2}^2})^2 z^2
-2 {\lambda_{12}^+}^4 
(x^2{\lambda_{12}^-}^4 - {\lambda_{1}^2} {\lambda_{2}^2}) 
 z
+2 K {\lambda_{1}} {\lambda_{2}} {\lambda_{12}^+}^2 \sqrt{z}-{\lambda_{12}^\ominus}^4}
{(x\lambda_{12}^{\ominus} +{\lambda_{1}} {\lambda_{2}})^2 (4  \lambda_{1}^2 \lambda_{2}^2 z^{\frac32} +
 K z )+4 \lambda_{1}^2 \lambda_{2}^2 {\lambda_{12}^\ominus}^2 \sqrt{z}+
K{\lambda_{12}^\ominus}^2 }, 
$$
$$
a_2 = 
\frac{4\lambda_2^2\lambda_1^2\sqrt{z} +K}
{{\lambda_{12}^\ominus}^2 + (x^2{\lambda_{12}^\ominus}^2-\lambda_1^2\lambda_2^2)z}, 
\qquad 
K = 16\lambda_2^4\lambda_1^4{z} - [{\lambda_{12}^\ominus}^2+ (x^2{\lambda_{12}^\ominus}^2 - \lambda_1^2\lambda_2^2)z]^2 .
$$

\begin{figure}[ht] 
\begin{center} 
\includegraphics[scale=0.32]{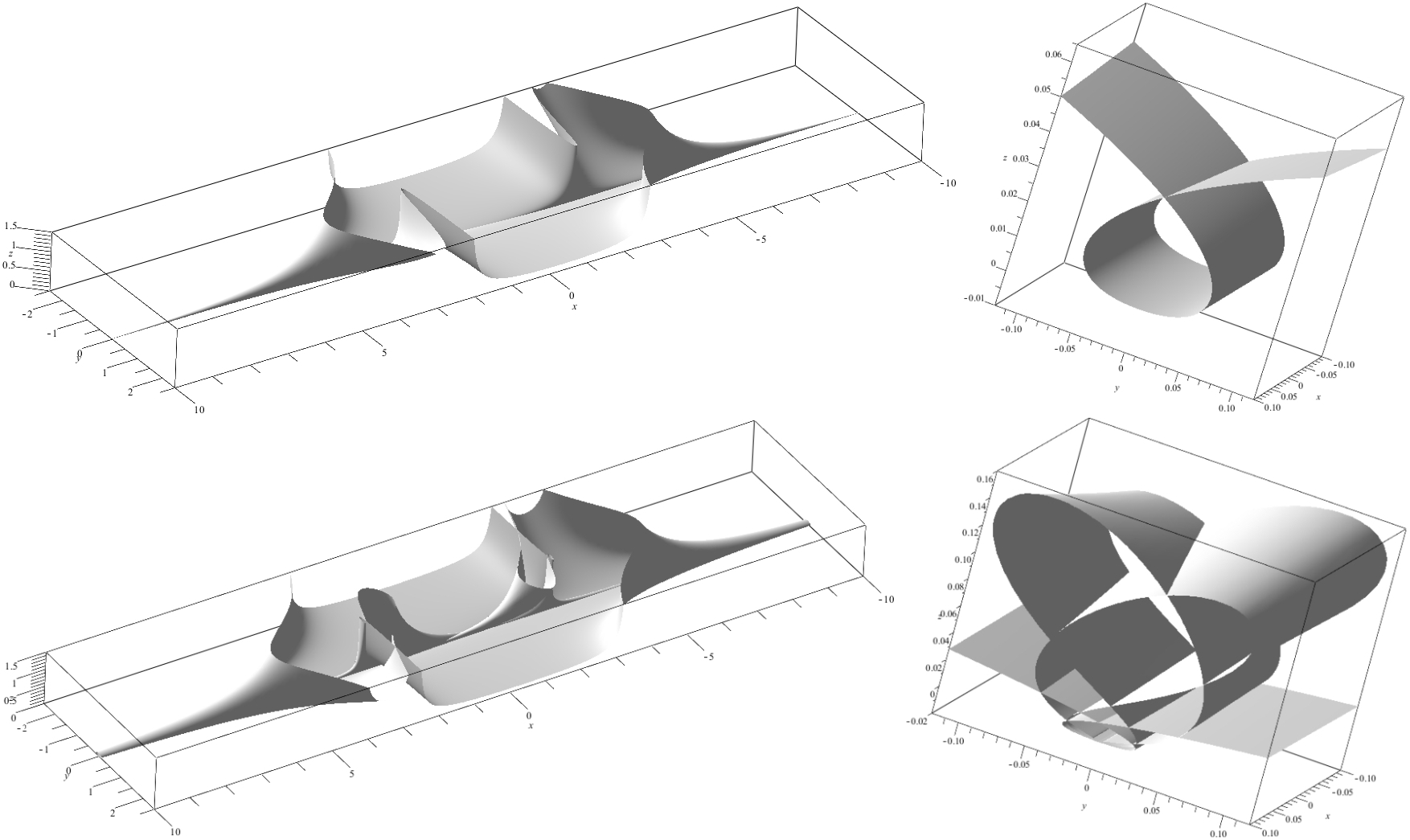}
\caption{Two soliton solution (top) and three soliton solution (bottom) of the CAE, $\lambda_1 = 1.2$, $\lambda_2 = 1.5$, $\lambda_3 = 1.8$, $c_i = 0$.  Right parts of the figures show the behavior around the origin.}\label{CAsol}
\end{center}
\end{figure}

Using \eqref{BTsurfseq} one obtains the two-soliton pseudospherical surface
$$
\mathbf{r}^{[2]} = 
\frac{4\lambda_{12}^\ominus}{\lambda_{1}^\oplus \lambda_{2}^\oplus R_1 }
[
(
\begin{array}{c}
R_2 \\
R_3 \\
0 
\end{array}
) \sin\alpha
+
(
\begin{array}{c}
R_3 \\
-R_2 \\
0
\end{array}
) \cos\alpha
]
+ 
(
\begin{array}{c}
0 \\
0 \\
\beta +  \frac{R_4}{\lambda_{1}^\oplus \lambda_{2}^\oplus R_1 }
\end{array}
) ,
$$
where
$$\numbered \label{P}
R_1  =   {\lambda_{12}^+}^2 (a_1^2+a_2^2) 
+  {\lambda_{12}^-}^2 (a_1^2 a_2^2+1)
-8 {\lambda_1} {\lambda_2} a_2 a_1, \\
R_2 = 2 {\lambda_1} {\lambda_2} (1+a_1 a_2) (a_1-a_2), \\
R_3  =  {\lambda_2}  {\lambda_1^\ominus} a_2 (1+a_1^2) - {\lambda_1} {\lambda_2^\ominus} a_1 (1+a_2^2), \\
R_4 =  2  {\lambda_{12}^\ominus} {\lambda_{12}^+} ({\lambda_1} {\lambda_2}-1) (a_1^2-a_2^2)  
- 2  {\lambda_{12}^\ominus} {\lambda_{12}^-} ({\lambda_1} {\lambda_2}+1) (a_1^2 a_2^2 - 1),  
$$
for picture see Figure \ref{PSsurf}.
\begin{figure}[ht] 
\begin{center} 
\includegraphics[scale=0.17]{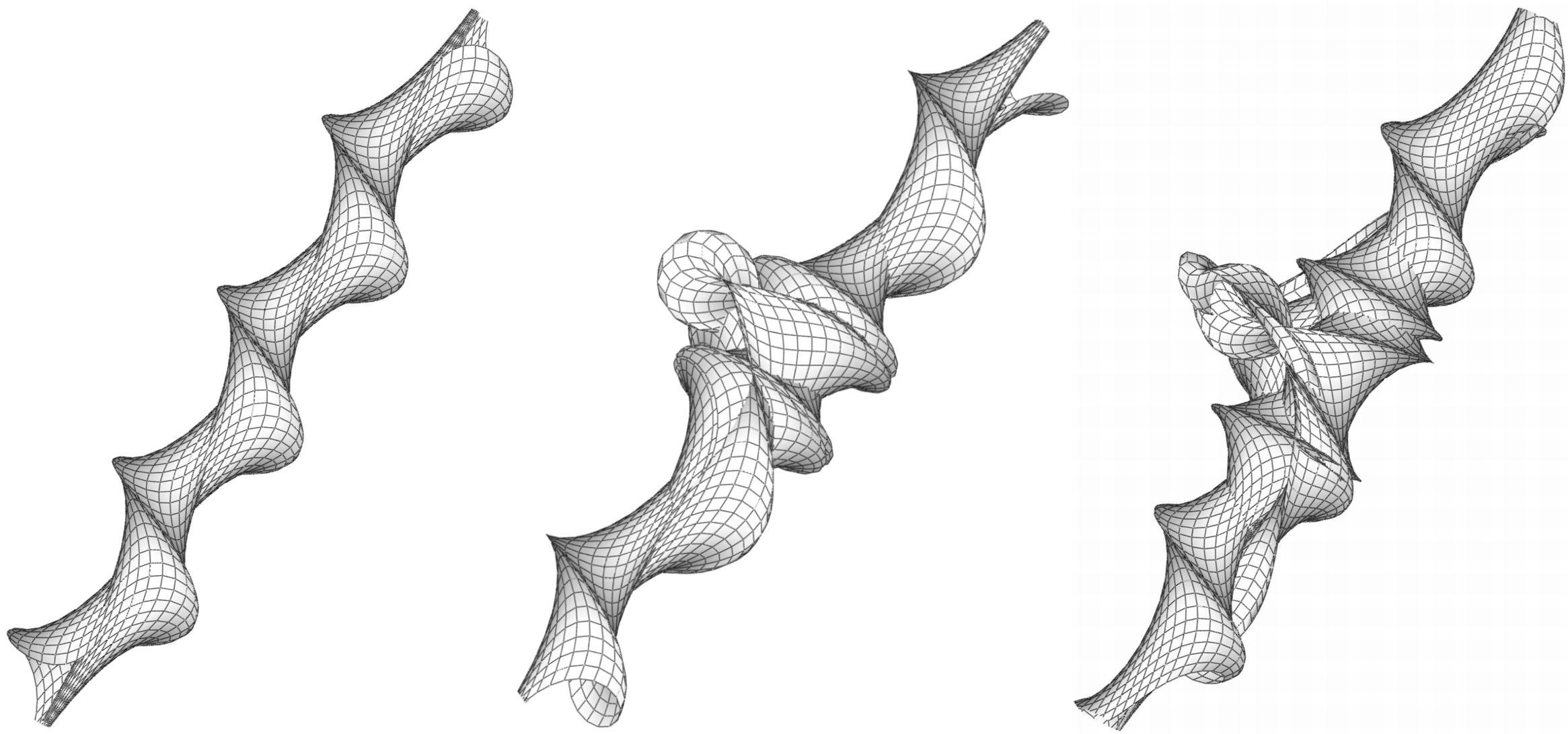}
\caption{From the left: One-, two- and three-soliton pseudospherical surfaces, $\bar{\mathbf r}^{[1]}, \mathbf{r}^{[2]}, \mathbf{r}^{[3]}$ respectively, 
$\lambda_1 = 1$, $\lambda_2 = 1.5$, $\lambda_3 = 1.8$, $c_i = 0$.}\label{PSsurf}
\end{center}
\end{figure}
The Gaussian map of corresponding constant astigmatism surface $\tilde{\mathbf{r}}^{[2]}$ is  
$$
\tilde{\mathbf{n}}^{[2]} = \mathbf{r}^{[2]}|_{\lambda_1 = 1} - \bar{\mathbf r}^{[1]} = 
(
\begin{array}{c}
N_1\sin\alpha + N_2\cos\alpha \\
N_2\sin\alpha - N_1\cos\alpha \\
N_3
\end{array}
), 
$$
where
$$
N_1 = -2\frac{{\lambda_2^\ominus} R_2'}{\lambda_2^\oplus R_1'}, \quad
N_2 = -2\frac{{\lambda_2^\ominus}R_3'}{\lambda_2^\oplus R_1'}
-\frac{4 a_2 \lambda_2}{{\lambda_{2}^\oplus} (a_2^2 + 1)}, \quad
N_3 = \frac{R_4'}{2 {\lambda_{2}^\oplus} R_1' } + \frac{2 \lambda_2 (a_2^2-1) }{{\lambda_{2}^\oplus} (a_2^2 + 1)}.
$$
Obviously, $R'_i$ arises from $R_i$ by substituting $\lambda_1 = 1$, namely
$$
\numbered 
R'_1  =   {\lambda_{2}^+}^2 (a^2+a_2^2) 
+  {\lambda_{2}^-}^2 (a^2 a_2^2+1)
-8 {\lambda_2} a a_2, \\
R'_2 = 2 {\lambda_2} (1+a a_2) (a-a_2), \\
R'_3 = - {\lambda_2^\ominus} a (1+a_2^2), \\
R'_4 =  -2  {\lambda_{2}^\ominus} {\lambda_{2}^+} ( {\lambda_2}-1) (a^2-a_2^2)  
- 2  {\lambda_{2}^\ominus} {\lambda_{2}^-} ( {\lambda_2}+1) (a^2 a_2^2 - 1). 
$$ 
Recall that $\tilde{\mathbf{n}}^{[2]}(\xi,\eta)$ parameterises the unit sphere by slip lines, 
for a picture see the left side of Figure~\ref{n2}. 
One can easily check that in the case when $c_i = 0$, the slip lines net is symmetric with respect to the point $\tilde{\mathbf{n}}^{[2]}(0,0) = (1,0,0)$, i.e. if 
$\tilde{\mathbf{n}}^{[2]}(\xi,\eta) = (n_x, n_y, n_z)$, then $\tilde{\mathbf{n}}^{[2]}(-\xi,-\eta) = (n_x, -n_y, -n_z )$. In the picture,
we observe the sphere from the positive part of the $x$-axis and the symmetry of the pattern can be clearly recognized.    
Some parts of the sphere are multiply covered and one can see singularities (envelopes) of slip lines.  
We leave open the question (suggested by a few numerical calculations) 
whether coordinate lines $\xi = \rm const$ and $\eta = \rm const$ converge to 
the poles, i.e. points $(0,0,\pm 1)$. 
\begin{figure}[ht] 
\begin{center} 
\includegraphics[scale=0.4]{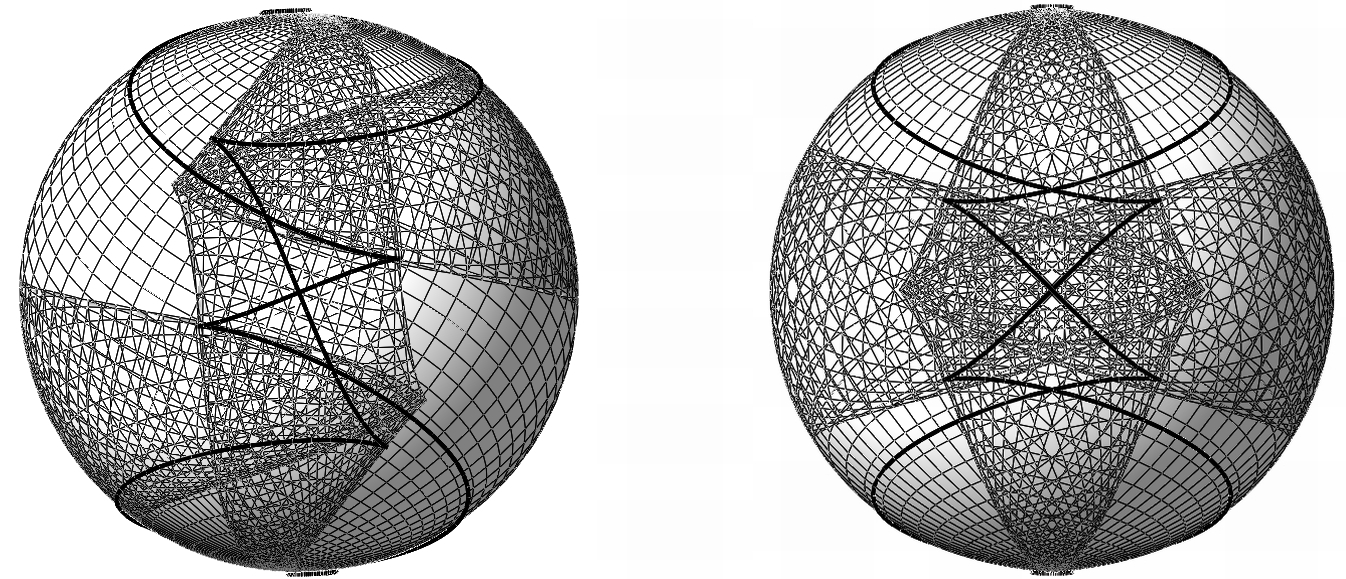}
\caption{Slip line field $\tilde{\mathbf{n}}^{[2]}$, $\lambda_2 = 1.5$, $c_i = 0$ (left) and its limit $\tilde{\mathbf{n}}^{\{2\}}$, $c = 0$ (right) with coordinate lines $\xi = 0$ and $\eta = 0$ highlighted (thick black curves).}\label{n2}
\end{center}
\end{figure}

To obtain an associated orthogonal equiareal pattern one needs to invert the transformation $(x,y) \leftrightarrow (\xi,\eta)$ given by first two equations in \eqref{CA2sol}, which is not possible in terms of elementary functions. The parameterisation 
$\tilde{\mathbf{n}}^{[2]}(x,y)$ then would provide orthogonal equiareal net sought. 

The corresponding family of two-soliton constant astigmatism surfaces having evolutes $\mathbf{r}^{[2]}|_{\lambda_1 = 1}$ and a Dini surface $\bar{\mathbf r}^{[1]} $ is then 
$$
\tilde{\mathbf{r}}^{[2]} = \bar{\mathbf r}^{[1]} - f \tilde{\mathbf{n}}^{[2]}  = 
(
\begin{array}{c}
\tilde{R}_1\sin\alpha + \tilde{R}_2\cos\alpha \\
\tilde{R}_2\sin\alpha - \tilde{R}_1\cos\alpha \\
\beta - \tilde{R}_3
\end{array}
), 
$$
where 
$$
f = \ln (k g^{[2]})|_{\lambda_1 = 1}, \quad
\tilde{R}_1 = 2 f \frac{{\lambda_2^\ominus} R_2'}{{\lambda_2^\oplus} R_1'}, \quad
\tilde{R}_2 = 2 f 
\frac{ {\lambda_2^\ominus} R_3'}{{\lambda_2^\oplus} R_1'} + 
\frac{4 \lambda_2 a_2}{{\lambda_2^\oplus} (a_2^2 + 1)}( f + 1   ), \\
\tilde{R}_3 =  \frac12 f \frac{R_4'}{{\lambda_2^\oplus}R_1'} 
+ \frac{2 \lambda_2 (a_2^2-1)}{{\lambda_2^\oplus} (a_2^2 + 1)}(f + 1).
$$ 
A part of the surface can be seen in the left side of Figure~\ref{CA23}. One can observe cuspidal edges  
obviously related to singularities of corresponding slip line field (Figure~\ref{n2}).
\begin{figure}[ht] 
\begin{center} 
\includegraphics[scale=0.5]{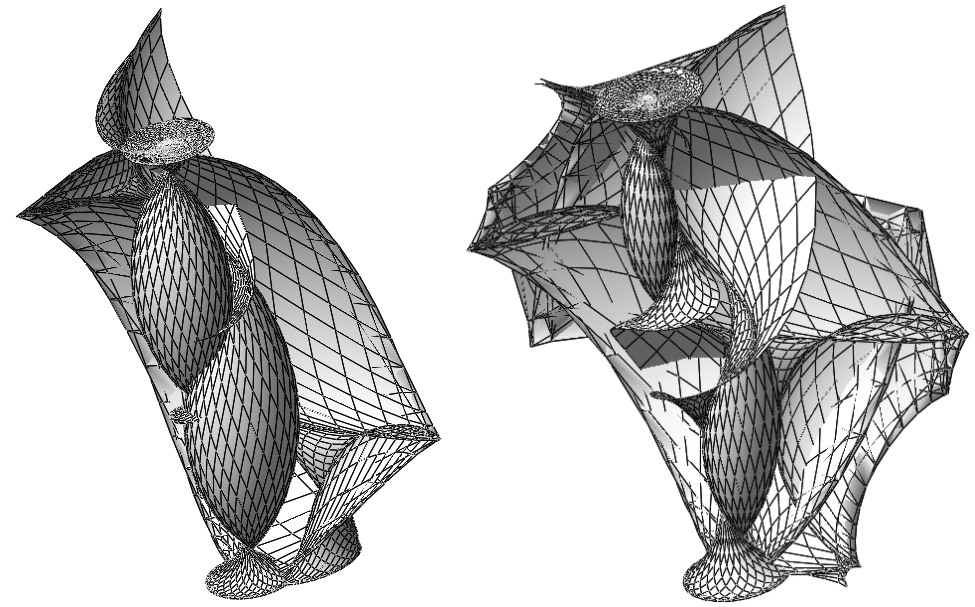}
\caption{Two soliton surface $\tilde{\mathbf{r}}^{[2]}$ (left) and three soliton surface $\tilde{\mathbf{r}}^{[3]}$ (right)  
of constant astigmatism, $\lambda_2 = 1.5$, $\lambda_3 = 1.8$, $c_i = 0$, $k = 1$.}\label{CA23}
\end{center}
\end{figure}

\subsection{Three-soliton solutions}
The 3-soliton solution, $(x^{[3]},y^{[3]},g^{[3]})$, 
of the CAE corresponding to the pair 
$$
\omega^{[3]}
= 2\arctan(
\frac{
  {\lambda_{12}^+} {\lambda_{13}^+} {\lambda_{23}^-} a_1
- {\lambda_{12}^+} {\lambda_{13}^-} {\lambda_{23}^+} a_2
+ {\lambda_{12}^-} {\lambda_{13}^+} {\lambda_{23}^+}  a_3
+ {\lambda_{12}^-} {\lambda_{13}^-} {\lambda_{23}^-}  a_1 a_2 a_3}
{
 {\lambda_{12}^-} {\lambda_{13}^+} {\lambda_{23}^+}  a_1 a_2
-{\lambda_{12}^+} {\lambda_{13}^-} {\lambda_{23}^+}  a_1 a_3
+{\lambda_{12}^+} {\lambda_{13}^+} {\lambda_{23}^-} a_2 a_3
+{\lambda_{12}^-} {\lambda_{13}^-} {\lambda_{23}^-} } 
) \\
$$ 
and $\bar{\omega}^{[2]}$ can be easily obtained using Proposition \ref{propSoliton}, 
explicit formulas being too lengthy to be written here. The graph of 
the solution $z(x,y) = z^{[3]}(x^{[3]},y^{[3]})$ can be seen in Figure \ref{CAsol},
a multivaluedness of the function $z$ being clearly identified. 
In the right side of the figure (values of $z$ near the point $(0,0,0)$) one can observe that at least eight values of $z$ 
may correspond to one particular choice of $x$ and $y$.

Proceeding to the constant astigmatism surface we firstly compute the three-soliton pseudospherical surface (see Figure \ref{PSsurf})
$$
\mathbf{r}^{[3]} = \mathbf{r}^{[2]} 
+ \frac{2\lambda_3}{\lambda_3^2 + 1}
(
\frac{\sin (\omega^{[2]} - \omega^{[3]})} {\sin 2 \omega^{[2]} } 
\mathbf{r}^{[2]}_\xi
+
\frac{\sin (\omega^{[2]} + \omega^{[3]})} {\sin 2 \omega^{[2]} } 
\mathbf{r}^{[2]}_\eta
).
$$
The Gaussian map of corresponding constant astigmatism surface is then
$$
\tilde{\mathbf{n}}^{[3]} = \mathbf{r}^{[3]}|_{\lambda_1 = 1} - \bar{\mathbf r}^{[2]}
$$
forming a net of slip lines, for picture see the left side of Figure~\ref{n3}. 
\begin{figure}[ht] 
\begin{center} 
\includegraphics[scale=0.4]{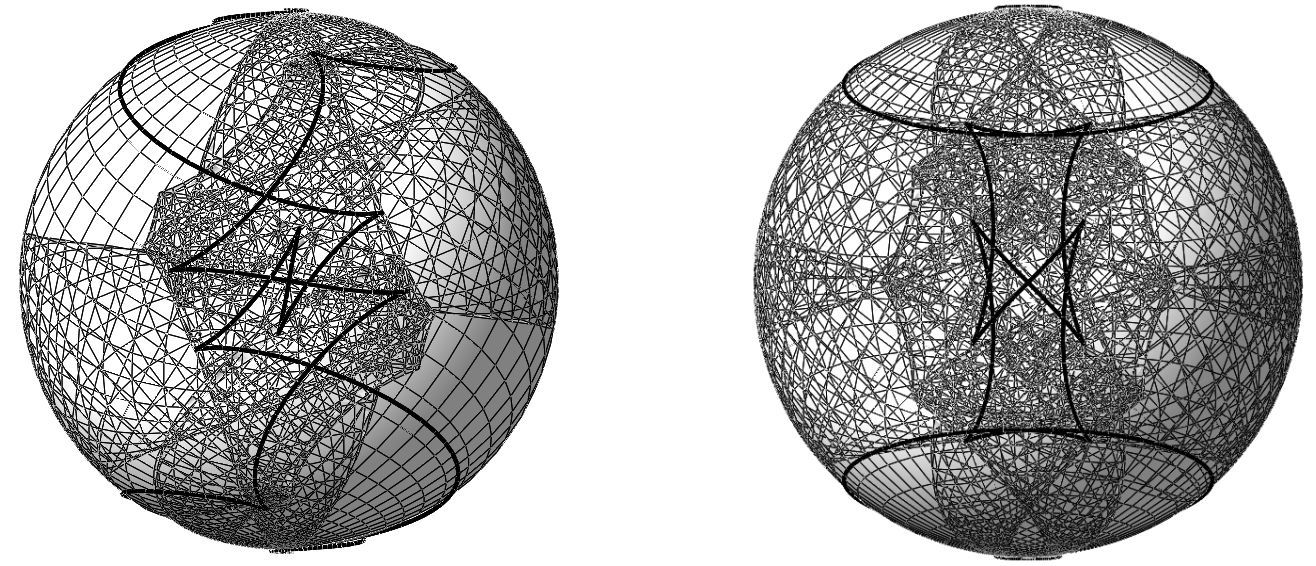}
\caption{Slip line field $\tilde{\mathbf{n}}^{[3]}$, $\lambda_2 = 1.5$, $\lambda_3 = 1.8$,  $c_i = 0$ (left) 
and its limit $\tilde{\mathbf{n}}^{\{3\}}$, $c = 0$ (right) with coordinate lines $\xi = 0$ and $\eta = 0$ highlighted (thick black curves).}\label{n3}
\end{center}
\end{figure} 

Finally, we construct the family of constant astigmatism surfaces having evolutes $\mathbf{r}^{[3]}|_{\lambda_1 = 1}$ and $\bar{\mathbf r}^{[2]} $. 
They are 
$$
\tilde{\mathbf{r}}^{[3]} = \bar{\mathbf r}^{[2]} -  \tilde{\mathbf{n}}^{[3]} \ln (k g^{[3]})|_{\lambda_1 = 1},
$$
$k$ being a constant. A picture of rather complicated surface can be seen in the right side of Figure~\ref{CA23}.

\section{Examples of multisoliton solutions with $\lambda = 1$}
Computing the limit for $\lambda_i \rightarrow 1$ allows us to iterate the B\"acklund transformation with $\lambda = 1$ and construct 
solutions of the sine-Gordon equation and corresponding pseudospherical surfaces. Slip-line fields are also available. 
However, taking the limit for $(\lambda_1,\lambda_{2}) \rightarrow (1,1)$ and $(c_1,c_2) \rightarrow (c,c)$ in \eqref{CA2sol} yields
$$
\lim x^{[2]} = \infty, \qquad \lim y^{[2]} = 0, \qquad \lim g^{[2]} = \infty.
$$
Fortunately, finite limits result when applying the scaling symmetry  
$$
\mathcal S_s (x,y,g) \rightarrow (sx, y/s, sg)
$$
of the CAE to the triple $(x^{[2]}, y^{[2]}, g^{[2]})$, the constant $s$ being suitably chosen.    
In what follows we iterate the B\"acklund transformation with $\lambda = 1$ (the case examined e.g. in \cite{S-B}) and we 
find the corresponding solutions of the CAE as the limits of the solutions obtained in previous section. 

Let us denote the $k$-th iteration of the B\"acklund transformation with 
 $\lambda = 1$ by $\omega^{\{k\}}$, according to the diagram 
\begin{displaymath}
    \xymatrix{ 
\omega^{\{0\}} \ar[r]^{1} 	& \omega^{\{1\}} \ar[r]^{1} & \omega^{\{2\}} \ar[r]^{1}  & \omega^{\{3\}} \ar[r]^{1} & \omega^{\{4\}} \ar[r]^{1} & \dots .
}
\end{displaymath}
The corresponding pseudospherical surface is, obviously, denoted by $\mathbf r^{\{k\}}$. We also extend the notation 
to the solutions of the constant astigmatism equation, i.e. $x^{\{k\}}, y^{\{k\}}, g^{\{k\}}$.

\subsection{One-soliton solutions}
Starting with the zero solution $\omega^{\{0\}} = 0$, the one-soliton solution of the sine-Gordon equation is 
$$
\omega^{\{1\}} = 2\arctan a = 
2\arctan \ex^{\xi + \eta + c}
$$
and the one-soliton pseudospherical surface, $\mathbf r^{\{1\}}$, is the pseudosphere \eqref{pseudosphere}.

The one-soliton solution of the CAE can be easily found by substituting $\lambda_1 = 1$ 
into Equations \eqref{CA1solxg} and \eqref{CA1soly}. 
Actually, since the solution does not depend on $\lambda_1$ (it disappears when eliminating parameters $\xi, \eta$ from 
\eqref{CA1solxg} and \eqref{CA1soly}), 
we obtain precisely the von~Lilienthal solution \eqref{vLsol}.

Corresponding constant astigmatism surfaces are exactly those given by \eqref{vLsurf}. 

\subsection{Two-soliton solutions}\label{2sollim}
Taking the limit of $\omega^{[2]}$ for $\lambda_i \rightarrow 1$ and $c_i \rightarrow c$ yields 
$$
\omega^{\{2\}} = 2\arctan\frac{2a\alpha}{1+a^2} = 
2\arctan\frac{2\ex^{\xi + \eta + c}(\xi - \eta)}{1+\ex^{2(\xi + \eta + c)}}.
$$
The corresponding two-soliton pseudospherical surface is
$$
\mathbf{r}^{\{2\}} =
\frac{1}{a^4+2(2 \alpha^2+1) a^2 + 1}
(
\begin{array}{c}
4a \alpha(1+a^2) \sin\alpha + 4 a (1+a^2) \cos\alpha
 \\
4 a (1+a^2) \sin\alpha  - 4 a \alpha (1+a^2) \cos\alpha
 \\
\beta (a^4+2(2 \alpha^2+1) a^2 + 1) - 2(a^4-1)
\end{array}
)
$$
and the slip-line field on the constant astigmatism surface's Gaussian sphere is
$$
\tilde{\mathbf{n}}^{\{2\}} = \mathbf{r}^{\{2\}} - \mathbf{r}^{\{1\}} = 
\frac{1}{a^6+(4 \alpha^2+3) (a^4+ a^2) + 1}
\\
\times 
(
\begin{array}{c}
4a(a^2+1)^2 \alpha\sin\alpha + 2a[a^4+ 2(1-2\alpha^2)a^2 + 1]\cos\alpha \\
2a[a^4+ 2(1-2\alpha^2)a^2 + 1] \sin\alpha - 4a(a^2+1)^2 \alpha\cos\alpha \\
(1-a^2)[a^4+2(1-2 \alpha^2) a^2 + 1]
\end{array}
), 
$$
see the right side of the Figure \ref{n2}. The picture indicates some symmetries of the pattern. Indeed, if
$c = 0$ and $\tilde{\mathbf{n}}^{\{2\}}(\xi,\eta) = (n_x, n_y, n_z)$, then
\begin{itemize}
\item $\tilde{\mathbf{n}}^{\{2\}}(-\xi,-\eta) = (n_x, -n_y, -n_z)$  
(symmetry with respect to the point $\tilde{\mathbf{n}}^{\{2\}}(0,0) = (1, 0, 0)$) ,
\item $\tilde{\mathbf{n}}^{\{2\}}(\eta,\xi) = (n_x, -n_y, n_z)$
(symmetry with respect to the zero meridian -- the image of the line $\xi = \eta$ ),
\item $\tilde{\mathbf{n}}^{\{2\}}(-\eta,-\xi) = (n_x, n_y, -n_z)$
(symmetry with respect to the equator -- the image of the line $\xi = -\eta$ ).
\end{itemize}
Coordinate lines of the slip line field $\tilde{\mathbf{n}}^{\{2\}}$ converge to the poles, namely 
$$
\lim_{\xi\rightarrow \pm \infty} \tilde{\mathbf{n}}^{\{2\}} = \lim_{\eta\rightarrow \pm \infty} \tilde{\mathbf{n}}^{\{2\}} = (0,0,\mp 1).
$$
According to the definition of slip lines, the corresponding orthogonal equiareal pattern has the same symmetries and the same
limit behavior. Unfortunately, no picture is provided; to do that one has to reparameterise the sphere which requires expressing 
$\xi, \eta$ in terms of $x^{\{2\}}, y^{\{2\}}$, i.e. inverting the transformation described by first two equations in \eqref{CA2sollim}. 
Plugging the functions $\xi(x,y), \eta(x,y)$ into $\tilde{\mathbf{n}}^{\{2\}}(\xi, \eta)$ 
then would provide the orthogonal equiareal net $\tilde{\mathbf{n}}^{\{2\}}(x,y)$.   

The two soliton solution $(x^{\{2\}}, y^{\{2\}}, g^{\{2\}}) $ of the CAE can be obtained by taking a limit for 
$\lambda_i \rightarrow 1$ and $c_i \rightarrow c$ of the 
triple $(x^{[2]}, y^{[2]}, g^{[2]}) $ scaled by the factor 
$s = -(\lambda_2 - \lambda_1)^2$, i.e.
$$
(x^{\{2\}}, y^{\{2\}}, g^{\{2\}}) = 
\mathop{\lim_{\lambda_i \rightarrow 1}}_{c_i \rightarrow c} \mathcal{S}_{-(\lambda_2 - \lambda_1)^2} (x^{[2]}, y^{[2]}, g^{[2]}).
$$
We have
$$\numbered\label{CA2sollim}
x^{\{2\}}  
= \frac{4 \ex^{2(\xi + \eta + c)}(\xi - \eta)}{\ex^{4(\xi + \eta + c)}+2[2(\xi-\eta)^2+1] \ex^{2(\xi + \eta + c)} +1}, \\
y^{\{2\}} 
= \frac{\ex^{2(\xi + \eta + c)} [\xi + \eta -(\xi - \eta)^2 + c]  + \xi + \eta + (\xi - \eta)^2 + c}{\ex^{2(\xi + \eta + c)}+1}, \\
g^{\{2\}} =  \frac{2 \ex^{\xi + \eta + c} (\ex^{2(\xi + \eta + c)}+1) }{\ex^{4(\xi + \eta + c)}+ 2[2 (\xi - \eta)^2+1] \ex^{2(\xi + \eta + c)} +1}.
$$
Then $z^{\{2\}} = 1/{g^{\{2\}}}^2$ as a function of $x^{\{2\}},y^{\{2\}}$ is a solution of the CAE.
Eliminating parameters $\xi,\eta$ 
provides the implicit formula for the function $z = z(x,y) = z^{\{2\}}(x^{\{2\}},y^{\{2\}})$, namely
$$
y = 
\frac{L^2-1}{L^2+1} 
(\frac{ x z}{ z x^2+1 })^2
- \ln{L} ,
\qquad
L = \frac{z x^2+1}{ 
\sqrt{ z  - (zx^2 + 1)^2  } 
- \sqrt{z}}.
$$
Finally, we write down the constant astigmatism surface $\tilde{\mathbf{r}}^{\{2\}}$ having evolutes $\mathbf{r}^{\{1\}} $ 
and $\mathbf{r}^{\{2\}}$.
For brevity, let us denote $\ln (k g^{\{2\}})$ by $f$. Then
$$
\tilde{\mathbf{r}}^{\{2\}} = \mathbf{r}^{\{1\}} - f\tilde{\mathbf{n}}^{\{2\}} 
= (
\begin{array}{c}
\tilde{R}_1\sin\alpha + \tilde{R}_2\cos\alpha \\
\tilde{R}_2\sin\alpha - \tilde{R}_1\cos\alpha \\
\beta + \tilde{R}_3
\end{array}
),
$$
where
$$
\tilde{R}_1 = -\frac{4 a(a^2+1) f \alpha }{a^4+2(2 \alpha^2+1) a^2 + 1}, \\
\tilde{R}_2 = -2a\frac{ (f-1) a^4-2[2 \alpha^2+(2 \alpha^2-1) f+1] a^2+f-1}
{a^6+(4 \alpha^2+3) (a^4+ a^2)+1}, \\
\tilde{R}_3 = \frac{(f-1) a^6 - [(4 \alpha^2-1) f+4 \alpha^2+1](a^4 - a^2)-f+1}
{a^6+(4 \alpha^2+3) (a^4+ a^2)+1}. 
$$
The results from this example coincide with ones obtainable using a 
\it reciprocal transformation \rm for the constant astigmatism equation introduced in \cite{H-M III}, 
albeit the parameterisation of the results is different. 
For instance, one can observe apparent similarity between the surface $\tilde{\mathbf{r}}^{\{2\}}$ 
in Figure \ref{CA2lim} and the surface \cite[Sect.~8, Ex.~5, Fig.~3]{H-M III}.  
\begin{figure}[ht] 
\begin{center} 
\includegraphics[scale=0.32]{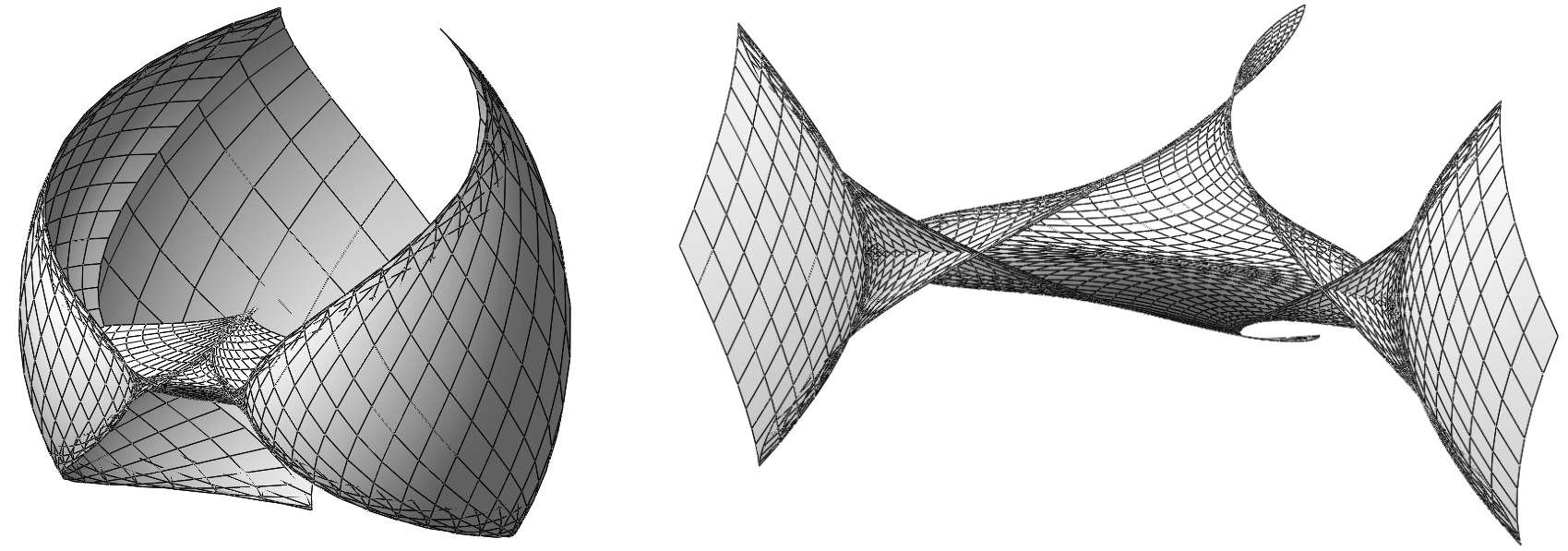}
\caption{Two soliton surface, $\tilde{\mathbf{r}}^{\{2\}}$, of constant astigmatism, $k = 1$, $c = 0$. 
The central part of the surface is zoomed in the right. }\label{CA2lim}
\end{center}
\end{figure}

\subsection{Three-soliton solutions}
The limit of $\omega^{[3]}$ for $\lambda_i \rightarrow 1$ and $c_i \rightarrow c$ yields
$$
\omega^{\{3\}} = 2\arctan\frac{a(2\alpha^2 + 2\beta+1) + a^3}{a^2(2\alpha^2 -2\beta + 1 ) + 1}
$$
or, alternatively \cite{S-B}, 
$$
\omega^{\{3\}} = - 2\arctan(\frac{1}{a}\cdot\frac  
{\beta+\frac12\sinh2(\beta+c)+\cosh^2(\beta+c)+\alpha^2}
{\beta+\frac12\sinh2(\beta+c)-\cosh^2(\beta+c)-\alpha^2}).
$$
The corresponding three-soliton pseudospherical surface is
$$
\mathbf{r}^{\{3\}} = \mathbf{r}^{\{2\}} 
+
(
\frac{\sin (\omega^{\{2\}} - \omega^{\{3\}})} {\sin 2 \omega^{\{2\}} } 
\mathbf{r}^{\{2\}}_\xi
+
\frac{\sin (\omega^{\{2\}} + \omega^{\{3\}})} {\sin 2 \omega^{\{2\}} } 
\mathbf{r}^{\{2\}}_\eta
)
$$
and the slip-line field on the constant astigmatism surface's Gaussian sphere is given by
$$
\tilde{\mathbf{n}}^{\{3\}} = \mathbf{r}^{\{3\}} - \mathbf{r}^{\{2\}},
$$
formulas being too lengthy to be written here, hence omitted. However, the picture can be seen in the right side of Figure~\ref{n3} and 
the properties of the pattern (symmetries when $c = 0$ and limits of the coordinate lines for $\xi,\eta \rightarrow \pm \infty$)
are exactly same as in the two soliton case from the previous subsection \ref{2sollim}.

In order to obtain the three soliton solution $(x^{\{3\}}, y^{\{3\}}, g^{\{3\}})$ 
we rescale the triple $(x^{[3]}, y^{[3]}, g^{[3]})$ by the scaling factor $s = (\lambda_2 - \lambda_1)^2(\lambda_3 - \lambda_1)^2$
and then compute the limit for $\lambda_i \rightarrow 1$ and $c_i \rightarrow c$, i.e.
$$
(x^{\{3\}}, y^{\{3\}}, g^{\{3\}}) = 
\mathop{\lim_{\lambda_i \rightarrow 1}}_{c_i \rightarrow c} \mathcal{S}_{(\lambda_2 - \lambda_1)^2(\lambda_3 - \lambda_1)^2} 
(x^{[3]}, y^{[3]}, g^{[3]}).
$$
The result is 
$$
x^{\{3\}} = 
  \frac{-4 a^2 [(\alpha^2-\beta) a^2-\alpha^2-\beta]} 
{a^6 + [4 (\alpha^2-\beta)^2 + 8 \alpha^2 + 3] a^4 + [4 (\alpha^2+\beta)^2 + 8 \alpha^2 + 3] a^2 + 1}, \\
y^{\{3\}} = \frac{2 \alpha}{3} \cdot 
\frac{ (\alpha^2-3 \beta+\frac32) a^4 - 2( \alpha^4-\alpha^2+3\beta^2 -\frac32) a^2+\alpha^2+3 \beta+\frac32}
{ a^4+2(2 \alpha^2+1) a^2 + 1}, \\
g^{\{3\}} =
\frac{2a[a^4+2(2 \alpha^2+1) a^2 + 1]}{a^6 + [4 (\alpha^2-\beta)^2 + 8 \alpha^2 + 3] a^4 + [4 (\alpha^2+\beta)^2 + 8 \alpha^2 + 3] a^2 + 1}.
$$

Finally, the constant astigmatism surface $\tilde{\mathbf{r}}^{\{3\}}$, part of which can be seen in Figure~\ref{CA3lim},
is
$$
\tilde{\mathbf{r}}^{\{3\}} = \mathbf{r}^{\{2\}} - \ln(k g^{\{3\}})\tilde{\mathbf{n}}^{\{3\}},
$$
$k$ being a real constant.
\begin{figure}[ht] 
\begin{center} 
\includegraphics[scale=0.38]{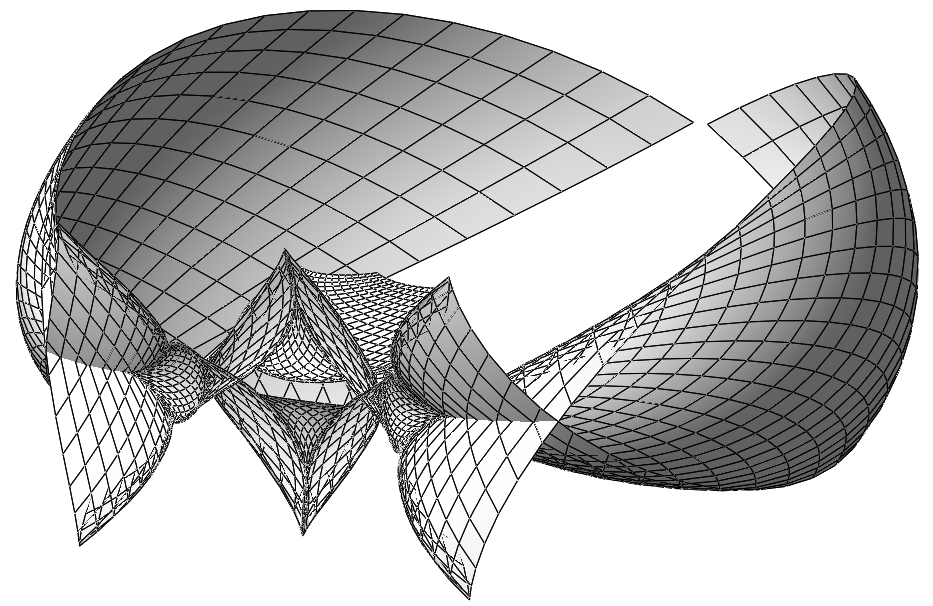}
\caption{Three soliton surface, $\tilde{\mathbf{r}}^{\{3\}}$, of constant astigmatism, $k = 1$, $c = 0$.}\label{CA3lim}
\end{center}
\end{figure}

\ack

The author was supported by Specific Research grant SGS/1/2013 of the Silesian University in Opava and
wishes to thank Michal Marvan for his guidance and valuable advice.
The author is indebted to Iosif S.~Krasil'shchik for reading the manuscript and
would also like to thank Petr Blaschke for useful discussions.

\end{document}